\documentclass[twocolumn,pra,showpacs,amsmath,amssymb]{revtex4}
\usepackage{graphicx}
\usepackage{amssymb,amsthm,amsfonts,amstext}
\usepackage{url}
\def\be{\begin{equation}}
\def\ee{\end{equation}}
\def\bea{\begin{eqnarray}}
\def\eea{\end{eqnarray}}
\def\bma{\begin{mathletters}}
\def\ema{\end{mathletters}}

\def\0{\overline{0}}

\def\q0{\underline{0}}

\def\H{{\cal H}}

\def\C{{\mathbb C}}
\def\id{{\mathbb I}}

\def\H{{\cal H}}

\def\R{\mathbb{R}}

\def\tr{\mbox{tr}}
\def\one{\leavevmode\hbox{\small1\normalsize\kern-.33em1}}

\def\bra#1{\langle#1|} \def\ket#1{|#1\rangle}

\def\proj#1{\ket{#1}\!\bra{#1}}

\newtheorem{theo}{Theorem}
\newtheorem{defin}[theo]{Definition}

\newtheorem{lemma}[theo]{Lemma}

\newtheorem{cor}[theo]{Corollary}
\def\id{{\mathbb I}}
\def\tr{\mbox{tr}}

\begin{document}

\title{The power of symmetric extensions for
entanglement detection}
\author{Miguel Navascu\'es, $^{1}$ 
        Masaki Owari,$^{1,2}$
        Martin B. Plenio,$^{1,2}$}
\address{$^1${\it Institute for Mathematical Sciences, 53 Prince's Gate,
Imperial College London, London SW7 2PG, UK}\\
$^2${\it  QOLS, Blackett Laboratory, Imperial College London, London SW7 2BW, UK }}

\begin{abstract}
In this paper, we present new progress on the study of the symmetric extension criterion for separability.
First, we show that a perturbation of order $O(1/N)$ is sufficient and, in general, necessary to destroy the entanglement of 
any state admitting an $N$ Bose symmetric extension. On the other hand, the minimum amount of local noise necessary to induce separability on states arising from $N$ Bose symmetric extensions with Positive Partial Transpose (PPT) decreases at least as fast as $O(1/N^2)$. From these results, we derive upper bounds on the time and space complexity of the weak membership problem of separability when attacked via algorithms that search for PPT symmetric extensions.
Finally, we show how to estimate the error we incur when we approximate the set of separable states by the set of (PPT) $N$-extendable quantum states in order to compute the maximum average fidelity in pure state estimation problems, the maximal output purity of quantum channels, and the geometric measure of entanglement.
\end{abstract}

\maketitle

\section{Introduction}
\label{intro}

The separability problem, that is, the problem to determine whether a given quantum state is separable or entangled, 
is one of the most fundamental problems in Entanglement Theory \cite{intro_measures}. Starting from the famous PPT (Positive Partial Transpose) criterion \cite{P96}, nowadays we have an enormous number of 
different separability criteria to choose from (see the citation lists 
of review papers in this topic \cite{intro_measures,B02,T02,SSLS05,I07,HHHH07,GT08}).    
Among all known separability criteria, those based on ``symmetric extensions'' and ``PPT symmetric extensions'' (i.e., symmetric extensions with an additional PPT constraint), as conceived by Doherty et al. \cite{doherty,doherty2},  
are considered to be among the most powerful \cite{I07}.
These criteria rely on the fact that any
set of $N$-symmetrically extendable states (PPT or not) converges
to a set of separable states in the limit of $N \rightarrow \infty$, as first noticed by Raggio and Werner \cite{raggio,werner}, although it also follows from the Quantum de Finetti theorem \cite{CFS02}.
Since both the set of $N$-symmetrically extendable states 
and the set of $N$-PPT symmetrically extendable states can be characterized by Semidefinite Programming \cite{sdp},
a well-known optimization problem for which many free solvers are available (like the MATLAB toolbox SeDuMi\cite{sedumi}),  
these tests are not only powerful, but also easy to implement.
This explains why, over all known numerical methods, the algorithms created by Doherty, Parrilo and Spedalieri (DPS) are the most popular in the Quantum Information community (notice, however, that there exist other methods for entanglement detection based on Semidefinite Programming besides the DPS criterion \cite{relaxations,brandao}).

This family of schemes has, though, an important drawback:
in this approach, in order to conclude that a given state $\rho$ is entangled, 
it is enough to find an $N$ such that $\rho$ does not belong to the set of $N$-(PPT) symmetric extendable states.
On the other hand, in order to show that a given state is separable, 
we would have to prove that it admits an $N$-(PPT) symmetric extension {\it for all natural numbers $N$}.
The DPS method then becomes useless: since we always operate under finite time and memory constraints, all we can do in practice is to check for the existence of $N$-(PPT) symmetric extensions for $N$ less or equal than some finite number $N_0$. If the state $\rho$ under analysis happened to admit an $N_0$ (PPT) symmetric extension, we could thus not conclude anything about its separability. 

Hulke and Bruss \cite{florian} tried to solve the issue by providing a complementary criterion designed to detect separability instead of entanglement, to be implemented at the same time as the DPS criterion. Unfortunately, the time complexity of that other method scales superexponentially with the dimension of the subsystems involved \cite{I07}. The reduced speed of convergence of the resulting two-way algorithm (much smaller than that of the DPS criterion) thus makes it unsuitable to study quantum correlations in high dimensional systems.

Besides, there is a more elegant way to approach the problem. 

In a recent work, Ioannou observed that, even if a state happens to have an $N_0$-(PPT) symmetric extension,
we can at least bound the distance between such state and the set of separable states in terms of $N_0$ \cite{I07}.
In the language of Computer Science, this means that the ``truncated'' DPS criterion allows to solve an instance of an approximate separability problem, the weak membership problem of separability ({\it WMEM}($\bar{S}$)). Ioannou therefore
provided an upper bound on the full time-complexity of the algorithm for WMEM based on 
symmetric extension criteria. 

But even after Ioannou's work, an open question remains to be solved.
The PPT symmetric extension criterion is considered to be stronger 
than the symmetric extension criterion \cite{doherty, doherty2}.
By definition, it is actually at least as strong as the symmetric extension criterion 
in the sense that a $N$-PPT symmetrically extendable state is $N$-symmetrically extendable. 
However, so far, there are no results that quantify {\it how strong} the additional PPT constraint makes the DPS criterion.
In particular, since the additional PPT constraint increases quadratically the size of the matrices that define the Semidefinite Programming problem,
there still remains the possibility that the PPT criterion just makes the DPS algorithm slower for {\it WMEM}($\bar{S}$).
In order to make this point clear, a similar analysis as Ioannou's should be done for the PPT-symmetric extension criteria.
Since Ioannou's analysis is based on the finite quantum de Finetti theorem \cite{KR05,CKMR07} and 
there exists no similar theorem for states satisfying the PPT constraint, there is no straightforward extension of Ioannou's work to the PPT symmetric extension criterion.

In this paper, by analyzing these criteria in more detail, 
we extend Ioannou's result to account for the PPT condition.

The structure of this article is as follows: in Section \ref{DPS} we will give the reader a detailed explanation of the DPS criterion and introduce the basic notation that will be used in the paper. Then we will move on to present the main result of this article, namely, an upper bound on the amount of noise needed to make the DPS states separable. This will allow us to compute upper bounds on the entanglement robustness of these states, and on their distance to the set of separable states. We will also briefly discuss how close our bounds are to being optimal. In Section \ref{sec: complexity}, we will use the previous results to analyze the computational complexity of solving the weak membership problem of separability through the DPS criterion. In particular, we will show that the PPT constraint in the DPS criterion reduces the dominant factor of the upper bound on the time complexity from $\left ( k_1 /\delta \right )^{6d_B}$ to $\left ( k_2/\delta \right )^{4d_B}$, where  $\delta$ is the accuracy parameter of {\it WMEM}($\bar{S}$). In Section \ref{sec: estimation} we will bound the speed of convergence of the DPS criterion when applied to compute the optimal fidelity in state estimation problems, the output purity of quantum channels and the geometric entanglement of arbitrary states. There we will perform some numerical tests to have a grasp at the actual speed of convergence of the DPS criterion, as opposed to our analytical upper bounds on it. In Sections \ref{prueba}, \ref{sec: multi} we will give the proof of the main theorem and explain how it can be extended to deal with the multipartite case. Afterwards, we will also show a very simple method to bound the entanglement of general PPT states. Finally, Section \ref{conclusion} will present our conclusions.

\section{The DPS criterion}
\label{DPS}

The Doherty-Parrilo-Spedalieri (DPS) criterion for entanglement detection \cite{doherty2} 
is a numerical algorithm that, combining the aforementioned results \cite{raggio,werner,CFS02} on $N$-extendibility with convex optimization methods, allows to characterize the set $S$ of separable operators up to arbitrary precision. 
The criterion arises from the following observation: if $\Lambda_{AB}\in S$, then, by definition, it belongs to the cone of bipartite product states, i.e.,

\be
\Lambda_{AB}=\sum_i p_i\proj{u_i}\otimes \proj{v_i},
\ee

\noindent with $p_i\geq 0$ for all $i$.

Once this decomposition is known, we can define a uniparametric family of operators $\Lambda_{AB^N}\in B(\H_A\otimes \H_B^{\otimes N})$ by tensoring $N$ times the last part:

\be
\Lambda_{AB^N}\equiv\sum_i p_i\proj{u_i}\otimes \proj{v_i}^{\otimes N}.
\label{exten}
\ee

Let us study the properties of the newly defined operators: first of all, from the above definition it is clear that they are all positive semidefinite. 
Also, from (\ref{exten}) it can be seen that tracing out the last $N-1$ systems we recover the initial operator, i.e., $\tr_{B^{N-1}}(\Lambda_{AB^N})=\Lambda_{AB}$, and that the last $N$ systems are invariant under the action of the permutation group. 
Finally, when viewed as an $N+1$-partite system, $\Lambda_{AB^N}$ is multiseparable, and therefore must remain positive semidefinite under the partial transposition of any bipartition of these systems.

For simplicity, we will incorporate all these properties in a single definition:

\begin{defin}{Bose symmetric extensions (BSE)} \\
Let $\Lambda_{AB}\in {\cal B}(\H_A\otimes\H_B)$ be a non-negative operator. 
We will say that $\Lambda_{AB^N}\in {\cal B}(\H_A\otimes\H_B^{\otimes N})$ is an $N$ Bose symmetric extension (BSE) of $\Lambda_{AB}$ iff:
\begin{enumerate}

\item $\Lambda_{AB^N}\geq 0$.

\item $\tr_{B^{N-1}}(\Lambda_{AB^N})=\Lambda_{AB}$.

\item $\Lambda_{AB^N}$ is Bose symmetric, i.e., $\Lambda_{AB^N}(\id_A\otimes P_{\mbox{sym}}^N)=\Lambda_{AB^N}$, where $P_{\mbox{sym}}^N$ denotes the symmetric projector of $N$ particles.
\end{enumerate}
\end{defin}

\noindent In case $\Lambda_{AB^N}$ is PPT with respect to all or some of its bipartitions $AB^K|B^{N-K}$, we will call it a \emph{PPT Bose symmetric extension} (PPT BSE) of $\Lambda_{AB}$.

From what we have seen, it is clear that, if $\Lambda_{AB}$ is a separable operator, then there exists an $N$ (PPT) BSE of $\Lambda_{AB}$ for any $N$. 
Since (PPT) Bose symmetric extensions are defined through linear matrix inequalities, the problem of determining whether a given state $\Lambda_{AB}$ admits one or not can be cast as a semidefinite program (SDP) \cite{sdp}, and therefore can be solved efficiently for fixed $N$ and varying dimensions. 
The DPS criterion consists precisely in, given an operator $\Lambda_{AB}$ whose separability is at stake, check for the existence of $N$ (PPT) Bose symmetric extensions for different values of $N$. 

A hierarchy of separability tests arises then naturally: if some operator $\Lambda_{AB}$ does not admit a (PPT) Bose symmetric extension for some $N$ (i.e., it does not pass the $N^{th}$ test), then it has to be entangled. 
If, on the contrary, such extension exists, then we would go for the $(N+1)^{th}$ test, that is, we would search for $N+1$ (PPT) Bose symmetric extensions of $\Lambda_{AB}$.
This last test would be in general more restrictive than the previous one, since for any $N+1$ (PPT) Bose symmetric extension $\Lambda_{AB^{N+1}}$ of $\Lambda_{AB}$ we can obtain an $N$ (PPT) Bose symmetric extension by tracing out the last system. 

Doherty et al. \cite{doherty} showed that the previous hierarchy completely characterizes the set of separable operators, in the sense that for any entangled positive operator $\Lambda_{AB}$ there exists an $N$ such that $\Lambda_{AB}$ does not admit an $N$ Bose symmetric extension.

We will now introduce a notation that will be used for the rest of the article: 
$S^N$ will denote the \emph{cone} of all bipartite operators that have an $N$ BSE, and $S_p^N$ will be understood as the set of all unnormalized quantum states that admit an $N$ BSE that is \emph{PPT with respect to the bipartition $AB^{\lceil N/2\rceil}|B^{\lfloor N/2\rfloor}$}. In case we also demand normalization, we will be dealing with the sets of states $\bar{S}^N, \bar{S}^N_p$. The elements of the previous four sets will be called $N$-(PPT) symmetrically extendable operators, or states, if normalized, or just DPS operators or states. Our previous discussion can then be summarized as

\begin{eqnarray}
& S^1\supset S^2\supset S^3\supset... \supset S,
& \nonumber\\ &S_p^1\supset S_p^2\supset S_p^3\supset... \supset S,& \nonumber\\
&\lim_{N\to\infty}S^N,S^N_p=S.&
\label{secuencia}
\end{eqnarray}

\noindent Note that $S^1=S^1_p$ $(\bar{S}^1=\bar{S}^1_p)$ is the set of all positive semidefinite operators (states).

Before ending this section, we would like to point out one additional fact.
As we already explained in the introduction, when we use the DPS criterion in practice,
it is not possible to conclude with certainty that a given state is separable. 
However, in the PPT case, 
by checking some rank constraints on the density matrices output by the computer,
we can \emph{sometimes} conclude separability in a finite number of steps. In that case, we will say that the PPT BSE presents a \emph{rank loop}.
We will make use of rank loops in Section \ref{sec: estimation} in order to estimate the accuracy of our upper bounds on the error we introduce when we perform linear optimizations over the sets $S^N$ or $S^N_p$ instead of $S$ in state estimation problems. A detailed explanation of this criterion for optimality can be found in Appendix \ref{optim}.

\section{Characterization of $S^N$ and $S_p^N$}
\label{charac}
We have seen that the sequences of sets $(S^N)$, $(S_p^N)$ tend to the set $S$ in the limit $N\to\infty$. Intuitively, this means that, for $N>>1$, any state $\rho_{AB}$ belonging to one of these sets must be either separable, or, at least, very close to a separable state. It seems thus plausible that the little entanglement such states may possess could be destroyed by some very attenuated local noise. One of the most simple noise models one can think of is depolarization, where a quantum state is turned into white noise with probability $p$. The action of the depolarizing channel $\Omega^{(p)}$ over some state $\rho\in B(\H)$ is given by

\be
\Omega^{(p)}(\rho)=(1-p)\rho+p\frac{\id}{d},
\label{depolarizing}
\ee

\noindent where $d$ is the dimension of the Hilbert space $\H$. Given any bipartite quantum state $\rho_{AB}$, shared by Alice and Bob, we could thus define its \emph{critical disentangling probability} $p_c(\rho_{AB})$ as the minimum probability with which one of the parties, say Bob, would have to prepare the maximally mixed state in his subsystem in order to disentangle it from Alice's. That is,

\be
p_c(\rho_{AB})=\min\{p:\id_A\otimes\Omega^{(p)}_B(\rho_{AB})\in \bar{S}\}.
\ee

\noindent Similarly, we can define the critical disentangling probability of a set of states $W$ as the maximum of all $p_c(\rho)$ for all $\rho\in W$. Clearly, $p_c\leq 1$ for all states, although this bound can be greatly improved if the dimensionality of Bob's system is small, as we shall see.

In this section, we will give upper bounds on this critical probability valid for
any state in $\bar{S}^N$ (or $\bar{S}_p^N$).
Then, by means of these results, 
we will provide several upper bounds on the speed of convergence of $\bar{S}^N$ and $\bar{S}_p^N$ to $\bar{S}$.

Before proceeding, though, a remark on notation: in this article, we will be mainly concerned with linear operators or quantum states acting over a bipartite Hilbert space $\H_A\otimes \H_B$, and all the formulas and bounds that we will derive in this section and the following three will involve the dimension of the Hilbert space $\H _B$ where the symmetric extensions are to be made. For the sake of clarity, we will therefore introduce the notation $d \stackrel{\rm def}{=} \dim \H_B$.

The following theorems will play a key role in deriving most of the results of this paper.

\begin{theo}
\label{bosesym}

\be
p_c(\bar{S}^N)\leq \frac{d}{N+d}.
\ee

\noindent In other words: for any operator $\Lambda_{AB}\in S^N$, the positive semidefinite operator

\begin{equation}
\tilde{\Lambda}_{AB}\equiv\frac{N}{N+d}\Lambda_{AB}+\frac{1}{N+d}\Lambda_A\otimes\id_B
\label{canonic1}
\end{equation}
 
\noindent is separable.

\end{theo}

\begin{theo}
\label{egregium}
Define $g_N$ (or $g_N^{(d)}$ in case $d$ is ambiguous) as

\begin{eqnarray}
g_N=& &\min\{1-x:P_{N/2+1}^{(d-2,0)}(x)=0\}\mbox{ for } N \mbox{ even},\nonumber\\
& &\min\{1-x:P_{(N+1)/2}^{(d-2,1)}(x)=0\}\mbox{ for } N \mbox{ odd},
\end{eqnarray}

\noindent with $P_n^{(\alpha,\beta)}(x)$ being the Jacobi Polynomials \cite{abramo}.

\noindent Then, 

\be
p_c(\bar{S}^N_p)\leq \frac{d}{2(d-1)}g_N.
\ee

\noindent That is, for any $\Lambda_{AB}\in S^N_p$, the positive semidefinite operator

\begin{equation}
\tilde{\Lambda}_{AB}\equiv(1-\frac{d}{2(d-1)}g_N)\Lambda_{AB}+\frac{1}{2(d-1)}g_N\Lambda_A\otimes\id_B
\label{canonic2}
\end{equation}
 
\noindent is separable.

\end{theo}
\noindent The proof of these two theorems is given in Section \ref{prueba},
where a separable decomposition for the states (\ref{canonic1}), (\ref{canonic2}) is also provided. Also, it is worth mentioning that, in both cases, $\tilde{\Lambda}_A=\Lambda_A$.

Notice that, in Theorem \ref{egregium}, $g_N$ is defined in terms of the greatest root of Jacobi polynomials.
The properties of the roots of Jacobi polynomials have been studied for quite time \cite{abramo}. 
This allows us to derive an expression for the asymptotic behavior of $g_N$: 
\begin{eqnarray}
g_N &\approx & 2\left(\frac{j_{d-2,1}}{N}\right)^2, \mbox{ for } N>>1,\\
& \approx &  2\left(\frac{d +1.856d^{1/3}+O(d^{-1/3}) }{N}\right)^2, \nonumber \\
&\quad & \qquad \mbox{ for } N\gg d\gg 1,
\label{bessel}
\end{eqnarray}
where $j_{n,1}$ is the first positive zero of the Bessel function $J_n(y)$.

How far can then the states in $\bar{S}^N,\bar{S}^N_p$ be from the set $\bar{S}$ of separable states? A way to answer this question could be to bound the maximum possible entanglement of such states.

The \emph{robustness of entanglement} of a state $\rho$ is defined as the minimum amount of separable noise needed to destroy the entanglement of such a state \cite{vidal}:
\begin{equation}
R(\rho) \stackrel{\rm def}{=}  
\min_{\lambda}\{\lambda :\exists \sigma \in \bar{S}, \ \mbox{s.t. } \frac{\rho+\lambda \sigma}{1+\lambda }\in S\}.
\end{equation}

\noindent The robustness of entanglement is also an upper bound on the \emph{global robustness of entanglement} $R_{G}(\rho)$ \cite{vidal}, defined by allowing $\sigma$ to be an arbitrary normalized quantum state in the above expression. And the global robustness of entanglement is, in turn, lower bounded by several other entanglement measures, like the negativity, the geometric measure of entanglement and the relative entropy of entanglement  \cite{vidal, computable,hayashishash,daniel,nature}. Any non trivial upper bound on the entanglement robustness of the states in $\bar{S}^N$ and $\bar{S}^N_p$ could thus retrieve a lot of information.

The following corollary follows straightforwardly from theorems \ref{bosesym} and \ref{egregium}.

\begin{cor}
Any $\rho \in \bar{S}^N$ satisfies
\begin{equation}
R(\rho ) \le \frac{d-1}{N} \label{hulk}.
\end{equation}
Similarly, any $\rho \in \bar{S}_p^N$ satisfies
\begin{equation}\label{hulkppt}
R(\rho) \le \frac{g_N}{2-\frac{d}{d-1}g_N} \approx \left(\frac{d}{N}\right)^2.
\end{equation}
\end{cor}

To see why, suppose that $\rho$ is normalized and use formulas (\ref{canonic1}), (\ref{canonic2}) to express $\tilde{\rho}$ (i.e., $\tilde{\Lambda}_{AB}$) in each case as a convex sum of the non negative operators $\rho$ and $\sigma\equiv\frac{1}{d-1}(\rho_{A}\otimes \id_B-\tilde{\rho})$. Then, notice that, since $\tilde{\rho}_A=\rho_A$ and $\tilde{\rho}$ is separable, then $\sigma$ must also be a separable operator \footnote{To understand why, write $\tilde{\rho}$ as a convex combination of product states, i.e., $\tilde{\rho}=\sum p_i\proj{u_i}\otimes\proj{v_i}$. Then, $\tilde{\rho}_A\otimes\id-\tilde{\rho}=\sum p_i\proj{u_i}\otimes(\id-\proj{v_i})$. That is, $\tilde{\rho}_A\otimes\id-\tilde{\rho}=\rho_A\otimes\id-\tilde{\rho}$ is a separable operator.}

Theorems \ref{bosesym} and \ref{egregium} also allow to obtain bounds on the distance between the states in $\rho_{AB} \in S^N,S^N_p$ and the set of separable states $\bar{S}$.

\begin{cor}\label{precisebounds}
For any $\rho \in \bar{S}^N$, there exist $\tilde{\rho} \in \bar{S}$ such that 
\begin{eqnarray}
&& \| \rho - \tilde{\rho} \|_{1} \le \frac{2(d-1)}{N+d-1}, \label{trace1}\\
&& \| \rho - \tilde{\rho} \|_{\infty} \le \frac{d-1}{N+d-1},\\
&& \| \rho - \tilde{\rho} \|_F = \frac{d}{N+d}\sqrt{\tr(\rho^2)-\frac{\tr(\rho_A^2)}{d}},
\end{eqnarray}
where $\| \cdot \|_{1}$, $\| \cdot \|_{\infty}$ and $\| \cdot \|_F$ are the trace, the operator and the Frobenius norm, respectively.

Similarly, for any $\rho \in \bar{S}_p^N$ (and $N\geq 2$), there exists a state $\tilde{\rho} \in \bar{S}$ such that 
\begin{eqnarray}
&& \| \rho - \tilde{\rho} \|_{1} \le g_N, \\
&& \| \rho - \tilde{\rho} \|_{\infty} \le g_N/2,\\
&& \| \rho - \tilde{\rho} \|_F = \frac{dg_N}{2d-2}\sqrt{\tr(\rho^2)-\frac{\tr(\rho_A^2)}{d}}.
\end{eqnarray}
\end{cor}
\begin{proof}
Here we give the proof for the bounds on the trace and operator norm. 
The proof for the Frobenius norm is omitted, since it is similar and simpler.
 
Let $\rho\in \bar{S}^N$. Them Theorem \ref{bosesym} implies that there exists $\tilde{\rho}\in \bar{S}$, with $\tilde{\rho}_A=\rho_A$, such that:

\be
\rho-\tilde{\rho}=\frac{d-1}{N+d-1}\rho-\frac{1}{N+d-1}(\rho_A\otimes\id_B-\tilde{\rho}).
\ee

\noindent Using the triangle inequality, we have that

\begin{eqnarray}
& &\|\rho-\tilde{\rho}\|_1\leq\frac{d-1}{N+d-1}\|\rho\|_1+\nonumber\\
& &+\frac{1}{N+d-1}\|(\rho_A\otimes\id_B-\tilde{\rho})\|_1=\frac{2(d-1)}{N+d-1},
\end{eqnarray}

\noindent where in the last step we used once more the fact that $\rho_A\otimes\id_B-\tilde{\rho}$ is separable (and, therefore, positive). Relation (\ref{trace1}) is thus proven.

For the operator norm, let $u_+(u_-)$ be the eigenvector corresponding to the maximum (minimum)
eigenvalue of $\rho-\tilde{\rho}$. It follows that

\be
\|\rho-\tilde{\rho}\|_\infty=\mbox{max}(\tr\{(\rho-\tilde{\rho})\proj{u_+}\},\tr\{(\tilde{\rho}-\rho)\proj{u_-}\}).
\ee

\noindent On the other hand,

\begin{eqnarray}
&\tr\{(\rho-\tilde{\rho})\proj{u_+}\}=\frac{d-1}{N+d-1}\tr\{\rho \proj{u_+}\}-\nonumber\\
&-\frac{1}{N+d-1}\tr\{(\rho_A\otimes\id_B-\tilde{\rho})\proj{u_+}\}\leq \frac{d-1}{N+d-1},
\end{eqnarray}

\noindent and

\begin{eqnarray}
&\tr\{(\tilde{\rho}-\rho)\proj{u_-}\}=-\frac{d-1}{N+d-1}\tr\{\rho \proj{u_-}\}+\nonumber\\
&+\frac{1}{N+d-1}\tr\{(\rho_A\otimes\id_B-\tilde{\rho})\proj{u_-}\}\leq \frac{d-1}{N+d-1}.
\end{eqnarray}

\noindent The first part of the corollary has been proven.

If $\rho\in \bar{S}^N_p$ and $N\geq 2$, then $\rho$ can be seen to be PPT. Since the PPT criterion implies
the reduction criterion \cite{reduction, reduction2}, we have that $\rho_A\otimes\id_B-\rho\geq 0$.
This observation, combined with the techniques used to derive the first set of relations, allows to prove the second one.
\end{proof}

The above corollaries can be reformulated as:
\begin{cor}\label{trnorm}
Suppose $\bar{S}(\delta)$ is a $\delta$-neighbor of the set of all separable states $\bar{S}$ in terms of the trace distance:
\begin{equation}
\bar{S}(\delta) \stackrel{\rm def}{=} \bigcup _{\rho \in \bar{S}} \left \{ \sigma \in \bar{S}^1 \ | \ \| \rho - \sigma\| \le \delta \right \}
\end{equation}
\noindent(remember that $\bar{S}^1$ is the set of all quantum states in $\H_A\otimes \H_B$). 

\noindent Then, the following relations hold:
\begin{eqnarray}
\bar{S}^N & \subset & \bar{S} \left( \frac{2(d-1)}{N+d-1} \right) \approx \bar{S}\left( 2\frac{d}{N}\right ), \label{eq:convergence sym} \\
\bar{S}_p^N & \subset & \bar{S}(g_N) \approx \bar{S}\left( 2\left( \frac{d}{N} \right)^2 \right),
\end{eqnarray}
\noindent where the approximations are granted to hold in the limit $N\gg d \gg 1$.
\end{cor}

\noindent This corollary suggests that the upper bounds for $\bar{S}_p^N$ converge quadratically faster than those for $\bar{S}^N$.
In other words, if these bounds were optimum, then we would have proven that the additional PPT constrain gives the DPS criterion a quadratic speed-up.

It is then natural to wonder if such bounds are indeed optimal. We will argue that at least the scaling of the upper bounds for $\bar{S}^N$ is correct, i.e., fixing $d_A$ and $d$, the maximum possible entanglement robustness of any bipartite state $\rho_{AB}$ arising 
from an $N$ Bose symmetric extension scales with $N$ as $O(1/N)$.

To see this, let $N=2K-1$, and consider the $N+1$ bipartite state given by

\be
\ket{\Psi_{AB^N}}\equiv\frac{1}{C_K}\sum _{\rm perm} \ket{\overbrace{0 \cdot 0}^{K} \overbrace{1 \cdot 1}^{K}},
\ee

\noindent where $C_K$ is a normalization factor. Define now $\rho_{AB}\equiv\tr_{B^{N-1}}(\proj{\Psi_{AB^N}})$. Clearly, $\rho_{AB}\in S^N$. Now, it can be shown that

\begin{eqnarray}
& &\rho_{AB}=\frac{K-1}{2(2K-1)}(\proj{00}+\proj{11})+\nonumber\\
& &+\frac{K}{2(2K-1)}(\ket{01}+\ket{10})(\bra{01}+\bra{10}).
\label{ejemplo}
\end{eqnarray}

\noindent The partially transposed operator $\rho_{AB}^{T_B}$ has a negative eigenvalue $-1/2(2K-1)$ 
corresponding to the eigenvector $(\ket{00}-\ket{11})/\sqrt{2}$, whose maximum Schmidt coefficient is $1/\sqrt{2}$. 
According to \cite{vidal}, this implies that $R(\rho_{AB})=1/(2K-1)=1/N$. 
The bound (\ref{hulk}) is, therefore, tight for $d_A=d=2$. 
Since for any pair for Hilbert spaces $\H_A,\H_B$ of dimensions greater than 1 we can embed the previous family of states in $B(\H_A\otimes \H_B)$, 
it follows that the optimal upper bound on the entanglement robustness of partial traces of Bose symmetric extensions must scale as $O(1/N)$.
On the other hand, the bound (\ref{hulkppt}) guarantees that the corresponding value for $\bar{S}_p^N$  
at least scales as $O(1/N^2)$, i.e., Theorem \ref{egregium} allows to derive an upper bound for the entanglement robustness that decreases asymptotically faster than the optimal upper bound in the general Bose symmetric case.

Note that the above considerations also allow us to obtain a dimension-dependent lower bound on the maximum possible entanglement robustness $R^N_{\sup}$ of a state in $\bar{S}^N$. Following the lines of \cite{shor}, consider the state $\sigma\equiv\rho_{AB}^{\otimes M}$, with $\rho_{AB}$ given by equation (\ref{ejemplo}). Clearly, $\sigma\in \bar{S}^N$, with $d_A=d_B=d=2^M$. As $-1/(2N)$ is the only negative eigenvalue of $\rho_{AB}^{T_B}$ and, therefore, the sum of its positive eigenvalues adds up to $1+1/(2N)$, the negativity of $\sigma$ \cite{neg1} (i.e., minus the sum of the negative eigenvalues of $\sigma^{T_B}$) can be seen equal to

\begin{eqnarray}
& &{\cal N}(\sigma)=\sum_{j=0}^{\lfloor (M-1)/2\rfloor}\left(\begin{array}{c}M\\2j+1\end{array}\right)\frac{\left(1+\frac{1}{2N}\right)^{M-2j-1}}{(2N)^{2j+1}}=\nonumber\\
& &=\frac{[(1+\frac{1}{2N})+\frac{1}{2N}]^M-[(1+\frac{1}{2N})-\frac{1}{2N}]^M}{2}=\nonumber\\
& &=\frac{(1+\frac{1}{N})^M-1}{2}\approx\frac{M}{2N},
\end{eqnarray}

\noindent where the last approximation is valid in the limit of large $N$. Since $R(\sigma)\geq {\cal N}(\sigma)$ \cite{computable}, it follows that $R(\sigma)\gtrapprox O(\log(d)/N)$. That is, for fixed dimension $d$, $R^N_{\sup}$ satisfies $O(\log(d)/N)\leq R^N_{\sup}\leq O(d/N)$.



\section{Computational complexity of WMEM($\bar{S}$)}\label{sec: complexity}
In this section, we will analyze the consequences of the previous results on separability from the point of view of Computer Science. 
Actually, there are several different ways to describe the separability problem as a computational problem \cite{I07}.
We chose to focus our attention in
an approximated separability problem called the \emph{weak membership problem of separability}.
This {\it ``promise''} problem (as opposed to a {\it ``decision'' } problem) roughly consists on
deciding the separability of a given state, but allowing an uncertainty parameterized by $\delta$.
In this Section we will derive upper bounds on the time and space complexity when we attack this problem via the DPS criterion.

The ``{\it In-biased}'' weak membership problem is defined as follows \cite{I07}: \\

\begin{defin}{Weak membership problem of separability (WMEM($\bar{S}$))} \\
 Given a bipartite quantum state $\rho \in \bar{S}^1$ and rational $\delta >0$, assert either that 
\begin{eqnarray}
\rho &\in& \bar{S}(\delta ) \  or \label{eq: wmem1}\\
\rho &\not\in& \bar{S},\label{eq: wmem2}
\end{eqnarray}

\noindent where $\bar{S}(\delta )$ is a $\delta$ neighbor of $\bar{S}$, i.e., $\bar{S}(\delta )=  
\{  \sigma \in \bar{S}^1 :\tilde{\sigma} \in \bar{S}\subset \bar{S}^1, \| \tilde{\sigma}-\sigma\|_1\leq \delta \}$. 

\end{defin}

In the above definition, $\| \omega \|_1=\tr(\sqrt{\omega\omega^\dagger})$, the trace norm of the operator $\omega$, 
although, in principle, we could have chosen other norms or distance measures as an accuracy parameter.

WMEM($S$) is, thus, an approximation of the conventional separability problem in the sense 
that an algorithm solving WMEM($\bar{S}$) may assert equation (\ref{eq: wmem1}) for a state $\rho_{AB}$ having just a small amount of entanglement.
This approximated formalism is more practical than a non-approximated or exact formalism like EXACT-QSEP \cite{I07}, because of
the inevitable errors we incur in both numerical and experimental studies, that should somehow be accounted for in our analysis of separability.
A fair amount of effort has been devoted to the study of the time complexity of WMEM($\bar{S}$), the most remarkable result being that,
if $d_A \geq d_B$, 
then WMEM($S$) is NP-hard whenever $1/\delta$ increases exponentially \cite{G03} or polynomially \cite{sevag} with respect to $d_B$.

We will now proceed to evaluate the time complexity of WMEM($\bar{S}$) when solved through the DPS criterion.
First, following the discussion of Doherty et al. \cite{doherty},
$S^N$ can be characterized by a semidefinite program with $\left ( ( \dim \H_{\mbox{sym}}^N )^2  - d_B^2 \right )d_A^2$ free variables
and a matrix of size $ ( \dim \H_{\mbox{sym}}^N ) d_A$ on which we will impose the positivity constraint. 
On the other hand, for $\bar{S}_p^N$, the PPT constraint implies demanding positivity from an additional matrix of size 
$ ( \dim \H_{\mbox{sym}}^{N/2} )^2 d_A$.
Since the time-complexity of an SDP with $m$ variables and of matrix size $n$ 
is $O(m^2n^2)$ (with a small extra cost coming from an iteration of algorithms),
the dominant factors for the asymptotic time-complexity of these tests can be written as
\begin{eqnarray}\label{eq: complexity}
    &{\rm Symmetric}:d_A^6(\dim \H_{\mbox{sym}}^{\overline{N_{\rm sym}}} )^6 \label{eq: complexity} \\ 
    &{\rm PPT\ symmetric}:d_A^6( \dim \H_{\mbox{sym}}^{\overline{N_{\rm ppt}}} )^4( \dim \H_{\mbox{sym}}^{\overline{N_{\rm ppt}}/2} )^4, 
\label{eq: complexity PPT} 
\end{eqnarray}
where $\overline{N_{\rm sym}}$ and $\overline{N_{\rm ppt}}$ are the sizes of the extensions 
needed to achieve a given accuracy parameter $\delta$.

Thus, at this stage, even though $\bar{S}^N_p$ converges to $\bar{S}$ faster than $\bar{S}^N$,
there still remains the possibility that the algorithm based on 
the sets $\{\bar{S}^N_p\}$ is slower than the one based on the sets $\{\bar{S}^N\}$, because of the increase in time complexity that arises from imposing positivity on the partially transposed operator.
The following calculation will rule out this possibility.

From  Eq. (\ref{eq:convergence sym}) of Corollary \ref{trnorm}, we have that

\begin{eqnarray}
&\overline{N_{\rm sym}} \leqq \frac{(2-\delta)(d_B-1)}{\delta},\nonumber\\
&\overline{N_{\rm ppt}}\lessapprox \frac{\sqrt{2}j_{d_B-2,1}}{\sqrt{\delta}}.
\end{eqnarray}

\noindent Taking into account that $j_{d,1}\approx d+O(d^{1/3})$ \cite{abramo}, 
the final expressions for upper bounds of the time complexity with respect to one method and the other are

\begin{eqnarray}
O\left(d_A^6\left[\frac{2e}{\delta}\right]^{6d_B}\right), & &\mbox{ for } \bar{S}^N \label{eq: complexity sym}\nonumber\\
O\left(d_A^6\left[\frac{e^2}{\delta}\right]^{4d_B}\right), & &\mbox{ for } \bar{S}_p^N, \label{eq: complexity ppt}
\end{eqnarray}
where we just wrote the dominant (exponential) terms and omitted all polynomially growing terms. Note that the scaling law derived for the non PPT DPS criterion is valid as long as the optimal bounds on the trace distance to the set of separable states scale as $d_B/N$. We conjecture that such is the case, although all our attempts to derive an analytical proof have failed so far. Under this assumption, the above formula thus shows that the criterion based on PPT BSEs indeed requires less steps than the one based on plain BSEs in order to solve WMEM($\bar{S}$) for a given accuracy $\delta$. 

The space complexity of both the plain DPS criterion and the PPT DPS criterion, though, is of the same type. This is because, although the PPT condition imposes (at least) a quadratic speedup in the speed of convergence, it also increases quadratically the size of the matrices involved in the SDP. Thus one effect cancels the other, and the size of the matrices needed in both cases to solve WMEM($\bar{S}$) up to a given precision $\delta$ is comparable for any value of $d_B$. It follows that, according to our bounds, in some situations it may be more convenient not to use the PPT condition in order to save memory space. 

Our experience with the DPS method suggests, however, that this expectation is not realistic, 
but rather a consequence of the non optimality of the bounds implicit in Theorem \ref{egregium}.
Actually, in practice, the algorithm based on PPT BSEs seems to have smaller space complexity than the one based on general BSEs.

A big underestimation of the role of the PPT condition in the DPS criterion could also explain why the bound (\ref{eq: complexity ppt}) behaves much worse than the asymptotic expressions $(k/\delta)^{2d_B}$ derived in \cite{I07} for the performance of the algorithm conceived by Ioannou et al. for entanglement detection \cite{witness,witness2}. Indeed, as we will see, our bounds on the distance between the sets $\bar{S}^N_p$ and the set of separable states are far from optimal, at least for small values of $d_A$. Therefore, a more refined analysis could in principle end up with a different scaling law for this distance, that would eventually lead to a much better estimate of the time complexity of methods based in PPT BSEs.





\section{Approximate algorithms for state estimation, maximum output purity, and geometric measure of entanglement} \label{sec: estimation}
There are many relevant quantities in quantum information whose definition involves a linear optimization over a set of separable operators. The maximum average fidelity in state estimation problems \cite{H76,H82},
the output purity of a quantum channel \cite{output} or the geometric measure of entanglement \cite{intro_measures} are examples of such quantities.
In order to compute these functions, we could think of an approximate algorithm that optimized over the sets $S^N$ or $S_p^N$
instead of $S$, and it is easy to see that such an algorithm would give the correct answer in the limit of large $N$. 

So far, we have seen how Theorems \ref{bosesym} and \ref{egregium} can be used to derive bounds related to the separability problem. In this Section we will show how to use these same theorems to bound the precision of the approximate linear optimizations over the cone of separable operators mentioned above.

\subsection{State Estimation Problems}

In a \emph{general} state estimation scenario, a source chooses
with probability $p_i$ a virtual quantum state $\Psi_i$ that is encoded
afterwards into another quantum state $\Psi'_i$, to which we are given
full access. The goal of the game is to measure our given state by means of a Positive Operator Valued Measure (POVM) $\{M_x\}_x$ and thus
obtain a classical value $x$ that we will use to make a guess $\phi_x$ on
the original state $\Psi_i$, which from now on we will assume to be pure.
In conventional estimation theory, we usually restrict the guess $\phi_x$ to be one of the original states $\{\Psi_i\}_i$ \cite{H76,H82}.
In this section, however, we will consider the more general setting in which we are allowed to choose arbitrary states
as a guess.

Being $\Psi_i$ a pure state, the efficiency of the protocol as a whole can be parametrized in terms of the \emph{average fidelity} $f$:

\begin{equation}
0\leq f\equiv\sum_{i,x}p_i\tr(\Psi_i'M_x)\tr(\phi_x\Psi_i)\leq 1.
\end{equation}

\noindent And the state estimation problem consists on determining $F$, the maximum fidelity among all possible measure-and-prepare schemes $(M_x,\phi_x)$.
Since $F$ can be used as well to determine whether a given quantum channel can be simulated or not
by an entanglement breaking channel, this problem is also referred to as the \emph{Quantum benchmark problem} \cite{HWPC2005,SDP07,AC08,CAMB08,OPPSW08}.

In \cite{pasado}, it is explained how to map the SE problem into a linear optimization over the set $S$ of separable states, via the relation

\begin{equation}
F=\max\{\tr(\rho_{AB}\Lambda_{AB}):\Lambda_{AB}\in S,\Lambda_A=\id\},
\label{funda}
\end{equation}

\noindent where $\rho_{AB}=\sum_i p_i \Psi'_i\otimes\Psi_i$ is given by the particular SE problem. 
There it is also shown that any separable decomposition of the optimal operator $\Lambda_{AB}=\sum_x M_x\otimes \phi_x$ corresponds to the optimal strategy $(M_x, \phi_x)$.

Now, consider the sequence of optimization problems:

\begin{eqnarray}
\hspace*{-1cm}& &F^N\equiv\max\{\tr(\rho_{AB}\Lambda_{AB}):\Lambda_{AB}\in
S^N ,\Lambda_A={\mathbb I}\},\nonumber\\
\hspace*{-1cm}& &F_p^N\equiv\max\{\tr(\rho_{AB}\Lambda_{AB}):\Lambda_{AB}\in
S^N_p ,\Lambda_A={\mathbb I}\},
\label{ideacentral}
\end{eqnarray}

\noindent From (\ref{secuencia}), it is immediate that $F^1\geq F^2\geq F^3\geq ... \geq F$, with $\lim_{N\to\infty}F^N=F$. An analogous property holds for the bounds $F^N_p$. Note that these maximizations are SDPs and therefore can be easily computed. 

Unfortunately, given limited computational (and specially memory) resources, it is only possible to compute these bounds up to some index $N$. 
In spite of the asymptotic convergence of the sequence, $F^N$ or $F^N_p$ could very well be far away from the actual solution of the problem. 
Is there any way to estimate the error of the truncation?

Take $\Lambda_{AB}\in S^N (S_p^N)$ to be the operator that maximizes equation (\ref{ideacentral}). 
Theorem \ref{bosesym} (\ref{egregium}) then implies that $\tilde{\Lambda}_{AB}$, as defined by equation (\ref{canonic1}) ((\ref{canonic2})), corresponds to a feasible state estimation strategy, since it is separable and $\tilde{\Lambda}_A=\Lambda_A=\id$. 
Moreover, we can use the separable decomposition of $\tilde{\Lambda}_{AB}$ that appears in Section \ref{prueba} to express it as a measure-and-prepare protocol $(M_x,\phi_x)$.

The fidelities $\tilde{F}^N$ or $\tilde{F}^N_p$ associated to these strategies, although non trivial, will not be optimal in general, but they should provide a lower bound for $F$. 
From (\ref{funda}), it is easy to see that

\begin{eqnarray}
& &\tilde{F}^N=\frac{N}{N+d}F^N+\frac{1}{N+d},\nonumber\\
& &\tilde{F}_p^N=\left(1-\frac{dg_N}{2(d-1)}\right)F_p^N+\frac{g_N}{2(d-1)}.
\end{eqnarray}

\noindent Notice that both lower bounds asymptotically converge to $F$. 
That is, from the solutions of the semidefinite programs (\ref{ideacentral}) it is possible to obtain a sequence of state estimation strategies that converges to the optimal measure-and-prepare scheme.

To have a grasp on the efficiency of the method, consider the following state estimation problem: suppose we have a device that outputs two copies of one of the 4 qubit states $\{\ket{\Psi_k}\}_{k=1}^4\equiv\{\ket{0},\ket{1},\ket{+},\ket{-}\}$ with equal probabilities. Our task is to estimate the state produced by the device. However, due to the environmental noise, once we are ready to measure the copies, those have degraded into $\rho_k\equiv \Omega^{(\epsilon)}(\proj{\Psi_k})=(1-\epsilon)\proj{\Psi_k}+\epsilon\id/2$. The results for $\epsilon=0.3$ are shown in Figure \ref{fidelitas}, for both the PPT and non PPT case and different values of $N$. 

\begin{figure}
  \centering
  \includegraphics[width=8 cm]{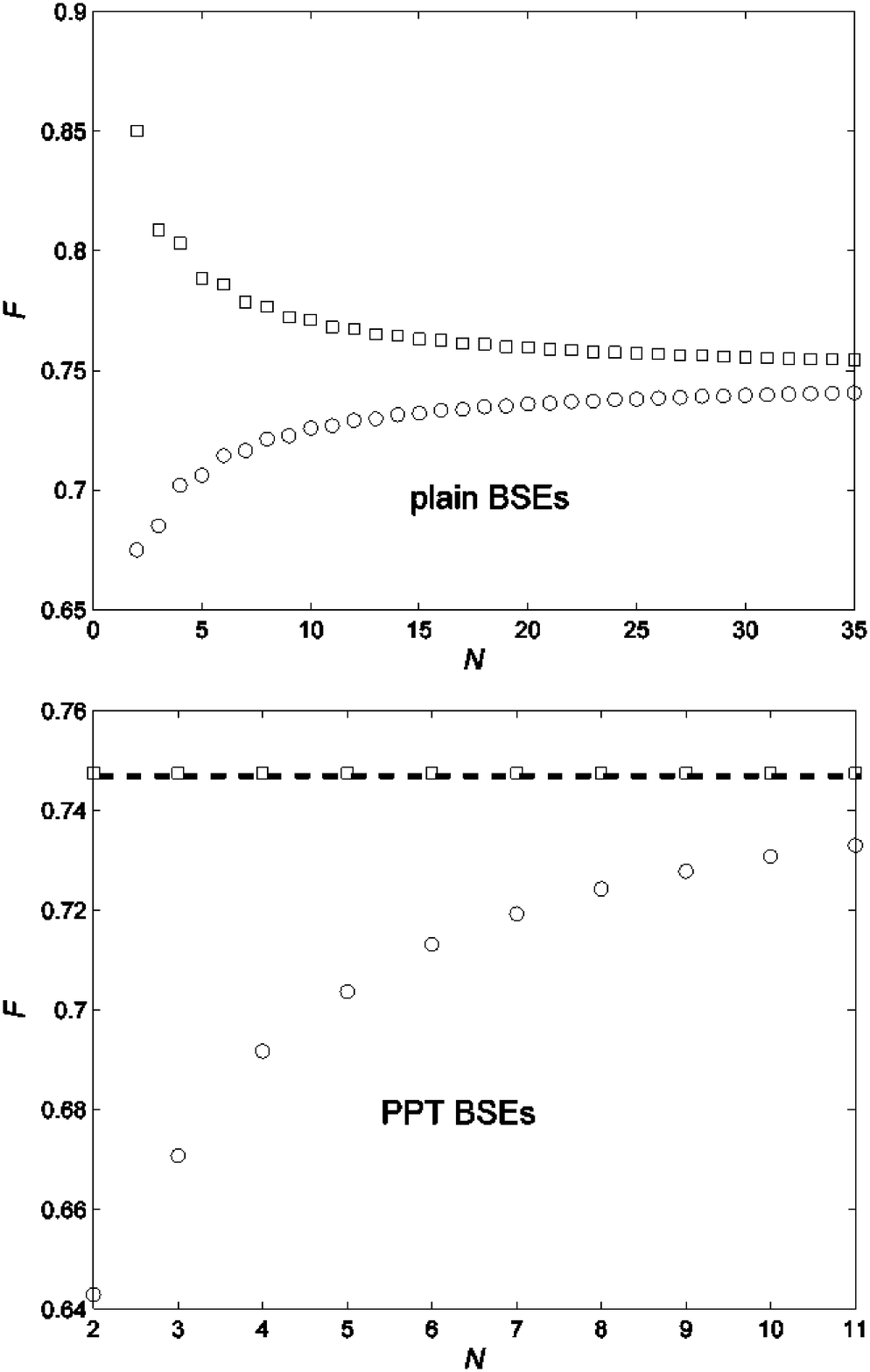}
  \caption{Upper (squares) and lower (circles) bounds for the maximum fidelity $F$ as a function of $N$. The dashed line indicates the value of the exact solution, attained exactly by the PPT upper bounds on $F$ from $N=2$ and onwards. The minimum difference between the upper and lower bounds is of the order of $10^{-2}$ in both plots.}
  \label{fidelitas}
\end{figure}

We used the MATLAB package \emph{YALMIP} \cite{yalmip} in combination with \emph{SeDuMi} \cite{sedumi} to perform the numerical calculations. Note that the curve corresponding to the upper bounds is constant, i.e., $F^N=F^M=F^*$, for all $M,N$. This suggested that $F^*$ could be equal to $F$, the solution of the problem, although we did not observe any rank loop in the matrices output by the computer. We thus had to \emph{force} the rank loop to occur. Using rank minimization heuristics \cite{lmirank} we checked for the existence of low rank PPT BSEs of $\Lambda_{AB}$ such that $\tr(\Lambda_{AB}\rho_{AB})\geq F^*-\delta$. Taking $\delta=10^{-4}$, the computer returned a matrix with a rank loop, therefore proving the optimality of $F^*$ up to this precision.

We performed a similar analysis for $d=3$, this time considering the problem where a degraded copy of one of the states

\begin{eqnarray}
\ket{\psi_{ij}}=& &\cos\left(\frac{j\pi}{6}\right)\ket{0}+\sin\left(\frac{j\pi}{6}\right)\cos\left(\frac{i\pi}{6}\right)\ket{1}+\nonumber\\
& &+\sin\left(\frac{j\pi}{6}\right)\sin\left(\frac{i\pi}{6}\right)\ket{2},
\end{eqnarray}

\noindent (where $i$ and $j$ run from 0 to 5) is sent to us with probability 1/36 through a depolarizing channel $\rho\to \Omega^{(0.2)}(\rho)$. In this case we were also able to force a rank loop in the PPT BSEs, so we again knew the optimal solution. Figure \ref{fidelitas2} illustrates our numerical results.

\begin{figure}
  \centering
  \includegraphics[width=8 cm]{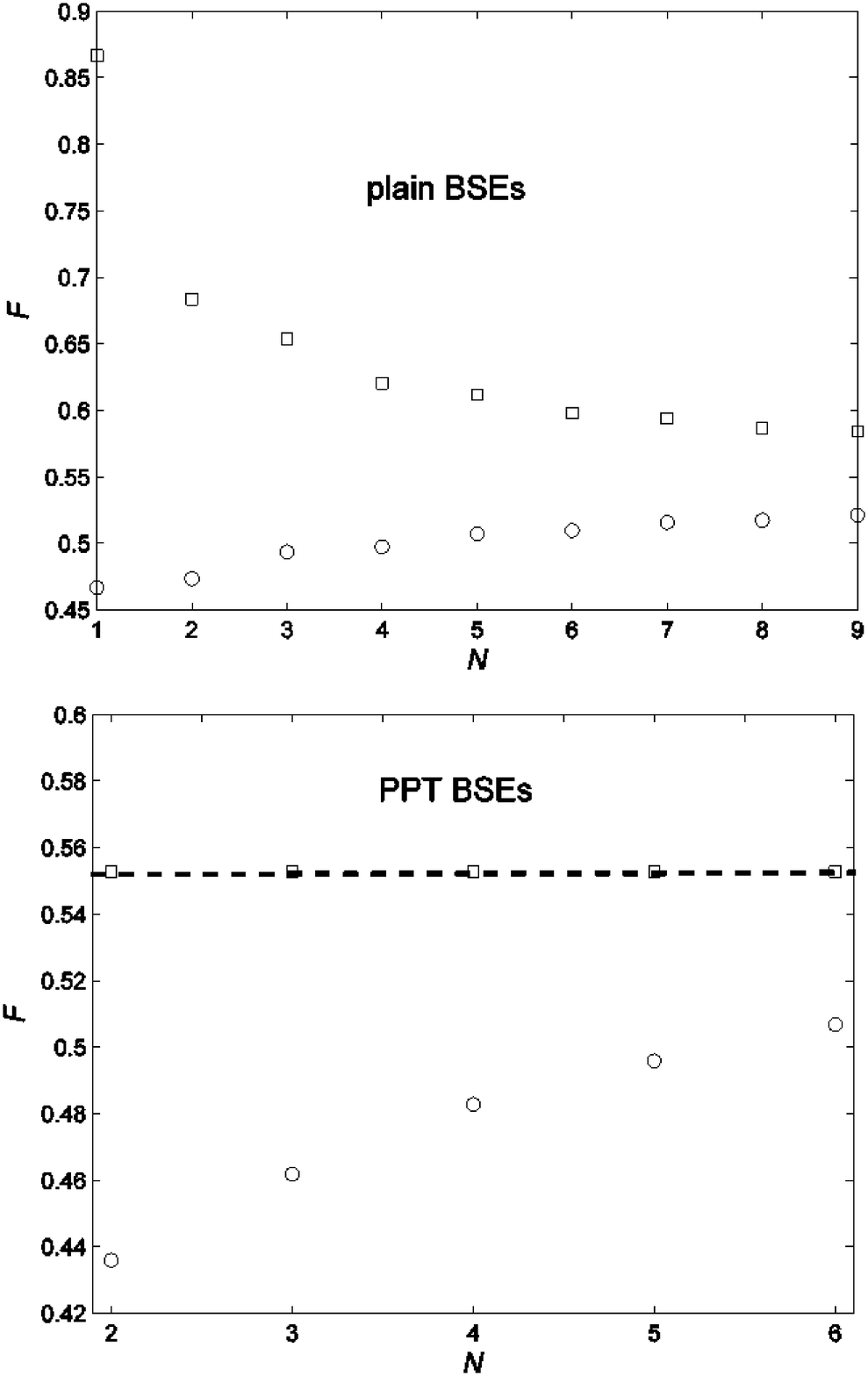}
  \caption{Upper (squares) and lower (circles) bounds for the maximum fidelity $F$ as a function of $N$ in dimension 3. This time, the minimum difference between our lower bounds and the exact solution is around 0.03 (and it is attained in the non PPT case).}
  \label{fidelitas2}
\end{figure}

Note that, in both cases, the lower bounds on the solution behave very similarly as the upper bounds given by the DPS criterion, as long as we are considering the non PPT case. In the PPT case, however, our bounds prove to be terrible, since the second available upper bound obtained through the DPS criterion already seems to attain the optimal solution. We will discuss briefly this topic in Section \ref{conclusion}.

The main features of the practical performance of the DPS criterion have already been illustrated above. Therefore, in the following two problems we will just stick to analytical results.

\subsection{Maximal output purity of quantum channels}

Let $\omega$ be a quantum channel. The \emph{maximal output purity} \cite{output} $\nu$ of $\omega$ is defined as

\be
\nu=\max_{\rho}\|\omega(\rho)\|_{\infty},
\label{output}
\ee

\noindent where the maximization is to be performed over all normalized quantum states $\rho$.

At first sight this quantity may seem extremely non linear. We will show that, actually, (\ref{output}) can be reformulated as a linear optimization over the set of separable states. 

Denote by $\Omega_{AB}$ the Choi operator corresponding to $\omega$, i.e., $\omega(\rho)=\tr_A(\Omega_{AB}\cdot\id_A\otimes\rho)$. It follows that

\be
\nu=\max_{\rho}\|\omega(\rho)\|_{\infty}=\max_{\rho,\sigma}\tr(\Omega_{AB}\cdot\sigma\otimes\rho),
\ee

\noindent with $\sigma, \rho\geq 0, \tr(\rho)=\tr(\sigma)=1$.

Or, equivalently,

\be
\nu=\max\{\tr(\Omega_{AB}\Lambda_{AB}):\Lambda_{AB}\in \bar{S}\}.
\ee

As in the state estimation case, it is possible to define decreasing sequences $(\nu^N)_N, (\nu_p^N)_N$ of upper bounds on $\nu$ that converge asymptotically to the optimal output purity of the channel. 
Using Theorems \ref{bosesym} and \ref{egregium}, together with the fact that $\tr_B(\Omega_{AB})=\id_A$, we have that there exist sequences $(\tilde{\nu}^N)_N, (\tilde{\nu}_p^N)_N$ of lower bounds on $\nu$ given by

\begin{eqnarray}
\tilde{\nu}^N=& &\frac{N}{N+d}\nu^N+\frac{1}{N+d},\nonumber\\
\tilde{\nu}^N_p=& &\left(1-\frac{dg_N}{2(d-1)}\right)\nu_p^N+\frac{g_N}{2(d-1)}.
\end{eqnarray}

\subsection{Geometric entanglement of tripartite pure states}

Let $\ket{\Psi}_{ABC}$ be a pure tripartite state. 
A popular entanglement measure for this kind of systems is the so called \emph{geometric entanglement} \cite{S95,geometric} (that some mathematicians may recognize as the square of the \emph{$\epsilon$-norm} \cite{tensor}), defined as

\be
E=\max_{\phi_A,\phi_B,\phi_C}|\bra{\phi_A}\bra{\phi_B}\bra{\phi_C}\ket{\Psi_{ABC}}|^2.
\ee

\noindent Notice, though, that, if we fix $\phi_A$ and $\phi_B$, the state $\phi_C$ maximizing the overlap will have to be proportional to $\bra{\phi_A}\bra{\phi_B}\ket{\Psi_{ABC}}$.
This overlap will be therefore equal to

\begin{eqnarray}
& &\tr_C(\bra{\phi_A}\bra{\phi_B}\ket{\Psi_{ABC}}\bra{\Psi_{ABC}}\ket{\phi_A}\ket{\phi_B}=\nonumber\\
& &=\tr(\rho_{AB}\proj{\phi_A}\otimes\proj{\phi_B}),
\end{eqnarray}

\noindent where $\rho_{AB}=\tr_C(\proj{\Psi_{ABC}})$. It follows that $E$ can also be reformulated as a linear optimization over $S$, i.e.,

\be
E=\max\{\tr(\Lambda_{AB}\rho_{AB}):\Lambda_{AB}\in \bar{S}\}.
\ee

As before, converging and decreasing sequences $(E^N)_N,(E^N_p)_N$ of upper bounds on $E$ can be derived via the DPS criterion, and Theorems \ref{bosesym}, \ref{egregium} allow us to obtain complementary increasing sequences of lower bounds $(\tilde{E}^N)_N,(\tilde{E}^N_p)_N$, given by

\begin{eqnarray}
\tilde{E}^N=& \frac{N}{N+d}E^N+\frac{1}{N+d}\lambda_{A},\nonumber\\
\tilde{E}^N_p=&  \left(1-\frac{dg_N}{2(d-1)}\right)E_p^N+\frac{g_N}{2(d-1)}\lambda_{A}.
\end{eqnarray}

\noindent Here $\lambda_{A}$ denotes the smallest eigenvalue of $\rho_A$.

\section{Proof of Theorems \ref{bosesym}, \ref{egregium}}
\label{prueba}
The purpose of this section is to derive Theorems \ref{bosesym}, \ref{egregium}. But first, a few words on notation.

Given a unitary operator $U$, by $\ket{U}$ we will denote the state $U\ket{0}$. Also, for any permutation $\pi\in P_N$, $V_\pi\in B(\H^{\otimes N})$ will represent the corresponding permutation operator. $V$ alone must be understood as the SWAP operator acting over a bipartite system $\H^{\otimes 2}$, i.e.,

\be
V=\sum_{i,j=0}^d \ket{i}\ket{j}\bra{j}\bra{i}.
\ee

\noindent To finish, $\H^N_{\mbox{sym}}$ will denote the symmetric subspace of $\H^{\otimes N}$ (the dimension of $\H$ will be clear from the context). 

We will now proceed to proof Theorems \ref{bosesym}, \ref{egregium}.

\noindent The basic idea for both proofs is to notice that the original problem of finding a separable state \footnote{Along this proof, we will consider the operator $\Lambda_{AB}$ to be a quantum state rather than a quantum operator, i.e., to be normalized. It is easy to see that, if (\ref{canonic1}) and (\ref{canonic2}) are separable states when $\Lambda_{AB}$ is normalized, the very same expressions have to lead to separable operators with $\tr(\tilde{\Lambda}_{AB})=\tr(\tilde{\Lambda}_{AB})$ if $\Lambda_{AB}$ is not normalized.} $\tilde{\Lambda}_{AB}$ very close to $\Lambda_{AB}$ from its BSE $\Lambda_{AB^N}$ can be viewed
as a \emph{probabilistic state estimation problem} \cite{fiurasek}.

Consider the following protocol, in which Alice plays a passive part:

\begin{enumerate}

\item A copy of $\Lambda_{AB^N}$ is distributed to two parties, Alice and Bob.

\item Bob performs performs an incomplete measurement over $\H_B^{\otimes N}$, described by the POVM $\{M_x\geq 0\}_x$, with $\sum_x M_x\leq \id$.
As a result, he obtains
either an outcome $x$ or a \emph{FAIL message}, indicating that his measurement
has failed to produce an outcome.

\item If Bob receives a FAIL message, then he makes it public. Otherwise, he prepares a state $\sigma_x\in B(\H_B)$, and both Alice and Bob would output the state $\frac{\tr_{B^N}(M_x\Lambda_{AB^N})\otimes\sigma_x}{p_x}$
with probability $p_x=\tr_{B^N}(M_x\Lambda_{B^N})$.

\end{enumerate}

The state Alice and Bob will produce conditioned on a non FAIL message will be then given by

\be
\tilde{\Lambda}_{AB}=\sum_x \frac{\tr_{B^N}(M_x\Lambda_{AB^N})\otimes\sigma_x}{\sum_y
p_y},
\ee

\noindent and is, therefore, a separable state. Moreover, since any entanglement breaking map can be decomposed as a measurement followed by the preparation of a state, this is the most general linear map we can apply over ${\cal H}_B^{\otimes N}$ in order to return a separable state $\tilde{\Lambda}_{AB}$.

But how to find a measure-and-prepare strategy for Bob such that $\tilde{\Lambda}_{AB}$ is close to $\Lambda_{AB}$? 
A possible scheme could be that Bob \emph{pretended} that his subsystems are $N$ identical copies of an unknown pure state, performed tomography over each of these subsystems independently and then prepared a state consistent with the average values he would measure. 
This strategy should give good results in the particular case where $\Lambda_{AB^N}$ can be approximated by a state of the form

\be
\int p(U)dU \rho_U\otimes \proj{U}^{\otimes N}.
\label{convexco}
\ee

However, supposing that the state had the form above, an even better strategy would be to allow Bob to perform \emph{collective} measurements over his subsystems and then prepare the most convenient state. 

In conclusion, Bob should apply a POVM that allows him to efficiently identify
the state $U\proj{0}U^\dagger$ out of $N$ copies of it. Because in
principle Bob has no a priori knowledge of $p(U)dU$, it is reasonable
that he assumes that $p(U)dU=dU$, the Haar measure. 

In this particular case,
the best state estimation strategy and the best probabilistic state estimation
strategy coincide \cite{fiurasek}. 
This implies that Bob should apply the POVM $\{\proj{U}^{\otimes N}dU\}_U$ and prepare the state $\proj{U}$ whenever he gets the result $U$. Therefore, 

\be
\tilde{\Lambda}_{AB}=\frac{\int dU \tr_{B^N}\left(\id_A\otimes \proj{U}^{\otimes N+1}\Lambda_{AB^N}\otimes
\id_B\right)}{\int dU \tr(\proj{U}^{\otimes N}\rho_{B^N})}.
\label{canonic}
\ee

\noindent To evaluate these integrals it is enough to notice that

\begin{enumerate}
\item For any operator $C$,

\be
\int dU U^{\otimes N}C(U^\dagger)^{\otimes N}=\sum_{\pi\in P_N} c_\pi
V_\pi,
\ee

\noindent for some coefficients $c_\pi$. In particular,

\begin{eqnarray}
& &\int dU \proj{U}^{\otimes N}=\frac{(d-1)!N!P_{\mbox{sym}}^N}{(N+d-1)!}=\nonumber\\
& &=\frac{(d-1)!\sum_{\pi\in
P_N}V_\pi}{(N+d-1)!}.
\end{eqnarray}

\item Due to the fact that $\Lambda_{AB^N}$ acts over ${\cal H}_A\otimes{\cal H}_{\mbox{sym}}^N$, for any $\pi\in P_{N+1}$,

\begin{eqnarray}
& &\tr_{B^N}\{(\Lambda_{AB^N}\otimes \id_B) \id_A\otimes V_\pi\}=\nonumber\\
& &=\Big\{\begin{array}{l}\Lambda_A\otimes \id_B,\mbox{ if } \pi(N+1)=N+1;\\
 \Lambda_{AB}, \mbox{ otherwise}.\end{array}
\end{eqnarray}

\end{enumerate}

\noindent Finally, we arrive at the expression

\be
\tilde{\Lambda}_{AB}=\frac{N}{N+d}\Lambda_{AB}+\frac{1}{N+d}\Lambda_A\otimes\id_B.
\ee

\noindent We have just proven Theorem \ref{bosesym}.

The next step is to extend the previous ideas to account for the PPT condition, and a possible way is to modify the previous bipartite protocol to give Bob the ability
to transpose part of his state before proceeding with any measure-and-prepare scheme. Suppose then that Bob partially transposes a partition $B'$,
corresponding to half of Bob's systems
in $\Lambda_{AB^N}$ (we will take $N$ even for simplicity). Following the previous arguments, Bob could pretend that he and Alice are sharing a state $\Lambda_{AB}^{T_{B'}}$
very similar to

\begin{equation}
\int p(U)dU \rho_U\otimes (\proj{U}\otimes \proj{U^*})^{\otimes N/2}.
\end{equation}

\noindent The benefits of this apparently useless step become evident when
we take into account the well established fact that it is easier to estimate a state from a copy and its complex conjugate than from two identical copies \cite{spinflip,fiurasek}. In the case of $N=2$, the optimal POVM has the form $\{U\otimes
U^*\proj{\varphi}(U\otimes U^*)^\dagger dU\}$, where $\ket{\varphi}$ is a linear combination
of $\ket{00}$ and $\ket{\Psi^+}=\sum_i \ket{ii}$, the (non normalized) maximally entangled state. The optimal
strategy for general $N$ is not known, but we suggest the measurement

\be
\phi_UdU\equiv(U\otimes U^*)^{\otimes N/2}\proj{\phi}(U^\dagger\otimes (U^*)^\dagger)^{\otimes N/2}
dU,
\ee

\noindent followed by the preparation of $\proj{U}$. 
Here $\ket{\phi}$ is an arbitrary linear combination of the states \footnote{Note that a further symmetrization of these states over the particles in $B'$ and $B\\B'$ would be more intuitive, but irrelevant, since the support of the state $\Lambda_{AB^N}^{T_{B'}}$ is in $\H_A\otimes\H_{\mbox{sym}}^{N/2}\otimes\H_{\mbox{sym}}^{N/2}$. 
That is, such symmetrization is automatically performed when we apply this POVM over $\Lambda_{AB^N}^{T_{B'}}$.} $\ket{\phi_n}\equiv\ket{00}^{\otimes n}\ket{\Psi^+}^{N/2-n}$, i.e.,

\be
\ket{\phi}=\sum_{n=0}^{N/2}c_n \ket{\phi_n}.
\ee

\noindent Of course, applying the POVM $\phi_U$ over $\Lambda_{AB}^{T_{B'}}$ is equivalent to apply
the (non positive!) map associated to $U^{\otimes N}\proj{\phi}^{T_{B'}}(U^\dagger)^{\otimes N/2}$ over our state $\Lambda_{AB^N}$. That way, we can use the same tricks
employed in the computation of (\ref{canonic}). 

A fast way to perform these calculations is to notice that, for $m>n$,

\small
\be
\ket{\phi_n}\bra{\phi_m}^{T_{B'}}=\proj{00}^{\otimes n}\otimes (\id\otimes\proj{0})^{\otimes m-n}\otimes V^{\otimes N/2-m}.
\ee
\normalsize

\noindent Therefore, there exists a pair of permutations $\pi,\pi'\in P_N$ such that

\be
V_\pi\ket{\phi_n}\bra{\phi_m}^{T_{B'}}V_{\pi'}^\dagger=\proj{0}^{\otimes m+n}\otimes \id^{\otimes N-m-n}.
\ee

But $\id_A\otimes V_\pi^\dagger \Lambda_{AB^N}= \Lambda_{AB^N}\id_A\otimes V_\pi=\Lambda_{AB^N}$, so 

\begin{eqnarray}
&\tr_{B^N}(\Lambda_{AB^N}\id_A\otimes U^{\otimes N}\ket{\phi_n}\bra{\phi_m}^{T_{B'}}(U^\dagger)^{\otimes N})=\nonumber\\
&=\tr_{B^N}(\Lambda_{AB^N}\id_A\otimes\proj{U}^{\otimes m+n}\otimes \id^{\otimes N-m-n}).
\end{eqnarray}

\noindent In the end, we have that

\be
\tilde{\Lambda}_{AB}=\left(1-d\frac{\vec{c}^\dagger\tilde{A}\vec{c}}{\vec{c}^\dagger\tilde{B}\vec{c}}\right)\Lambda_{AB}+\frac{\vec{c}^\dagger\tilde{A}\vec{c}}{\vec{c}^\dagger\tilde{B}\vec{c}}\Lambda_A\otimes\id_B,
\label{PPT}
\ee

\noindent where $\tilde{A}$ and $\tilde{B}$ are square matrices given by

\begin{eqnarray}
&\tilde{B}_{nm}=\frac{(n+m)!}{(n+m+d-1)!},\tilde{A}_{nm}=\frac{(n+m)!}{(n+m+d)!},\nonumber\\
&n,m=0,1,...,N/2.
\end{eqnarray}

In case of odd $N$, we would make Bob partially transpose $(N-1)/2$ parts
of his state and then use the following (incomplete) POVM:

\be
U^{\otimes N}\proj{\phi}^{T_{B'}}\otimes\proj{0}(U^\dagger)^{\otimes N}dU.
\ee

\noindent After the appropriate computations, we again arrive at expression
(\ref{PPT}), but the form of $\tilde{A}$ and $\tilde{B}$ changes to:

\begin{eqnarray}
&\tilde{B}_{nm}=\frac{(n+m+1)!}{(n+m+d)!},\tilde{A}_{nm}=\frac{(n+m+1)!}{(n+m+d+1)!},\nonumber\\
&n,m=0,1,...,(N-1)/2.
\end{eqnarray}

Obviously, in order to guarantee that $\Lambda_{AB}$ is close to $\tilde{\Lambda}_{AB}$,
it is in our interest to minimize the quantity 

\be
f_N(\vec{c})\equiv\frac{\vec{c}^\dagger\tilde{A}\vec{c}}{\vec{c}^\dagger\tilde{B}\vec{c}}
\label{complicaciones}
\ee

\noindent over all possible vectors $\vec{c}$. Details on how to calculate the minimum of (\ref{complicaciones}), together with the expression of the
optimal $\vec{c}$ can be found in Appendix \ref{minimum}. The result is:

\be
\min_{\vec{c}}f_N(\vec{c})=\frac{1}{2(d-1)}g_N.
\ee

\noindent This concludes the proof of Theorem \ref{egregium}.

Notice that in both cases the given separable decomposition of the states $\tilde{\Lambda}_{AB}$ is continuous. Because of the presence of the Haar measure, however, via Design Theory it is possible to arrive at an approximate \cite{emerson} or exact \cite{hayashi} finite separable decomposition for these operators.


\section{Extensions to multiseparability} \label{sec: multi}
So far, we have only been considering separability in \emph{bipartite} systems.
In this section, we show that almost all the results we have derived  
can be easily extended to deal with separability in $m$-partite scenarios.
More concretely, we will show how to generalize Theorems \ref{bosesym} and \ref{egregium} to the multipartite case, since, as we have already seen, most of the other results are just corollaries of these two theorems.

In this case, we will be interested in sets $S^N$ of states that derive from an $N$ \emph{locally} (PPT) Bose-symmetric extension\cite{trisep}.

\begin{defin}{$N$ locally Bose-symmetric extension}\\
Let $\Lambda_{123...}\in {\cal B}(\H_1\otimes\H_2\otimes\H_3\otimes...)$ be a non negative operator. We will say that $\Lambda_{12^{N}3^{N}...}\in {\cal B}(\H_1\otimes\H_2^{\otimes N}\otimes\H_3^{\otimes N}\otimes...)$ is an $N$ locally Bose symmetric extension of $\Lambda_{123...}$ iff:
\begin{enumerate}

\item $\Lambda_{12^{N}3^{N}...}\geq 0$.

\item $\tr_{2^{N-1}3^{N-1}...}(\Lambda_{12^{N}3^{N}...})=\Lambda_{123...}$.

\item $\Lambda_{12^N3^N...}$ is independently Bose symmetric in systems $2,3,4...$.
\end{enumerate}
\end{defin}

As before, in case such extension is PPT with respect to some partition, we will denote it as an \emph{$N$ PPT locally Bose-symmetric extension}.

How close is $\Lambda_{123...}$ to the set of separable states? Consider a triseparable system, for instance, and suppose that we have an $N$ locally Bose-symmetric extension $\Lambda_{AB^{N}C^{N}}$ for $\Lambda_{ABC}$. In order to estimate the distance of $\Lambda_{ABC}$ to the set of triseparable states we could conceive a protocol where the state $\Lambda_{AB^{N}C^{N}}$ is distributed between Alice, Bob and Charlie. As before, Bob and Charlie could then independently apply probabilistic state estimation over their subsystems and prepare both a quantum state depending on their measurement outcomes.

From what we already have, the derivation of the final expression of the triseparable state $\tilde{\Lambda}_{ABC}$ is straightforward. Equation (\ref{canonic1}) describes the action of Bob's strategy over \emph{any} bipartite state. Considering the partition $AC^N|B^N$, it follows that the resulting tripartite state after Bob performs state estimation will be:

\be
\frac{N}{N+d_B}\Lambda_{ABC^N}+\frac{1}{N+d_B}\Lambda_{AC^N}\otimes\id_B.
\ee

Now it is Charlie's turn. This time we will take the partition $AB|C^N$. The final result is that

\small

\begin{eqnarray}
&\tilde{\Lambda}_{ABC}=\frac{N^2}{(N+d_B)(N+d_C)}\Lambda_{ABC}+\frac{N}{(N+d_B)(N+d_C)}\Lambda_{AB}\otimes\id_C+\nonumber\\
&+\frac{N}{(N+d_B)(N+d_C)}\Lambda_{AC}\otimes\id_B+\frac{1}{(N+d_B)(N+d_C)}\Lambda_{A}\otimes\id_{BC}
\end{eqnarray}

\normalsize

\noindent is a triseparable state.

The generalization to more parties is immediate. Invoking again the definition of depolarizing channels (\ref{depolarizing}), in $m$-partite separability the expression for $\tilde{\Lambda}_{1234...}$ would be

\be
\tilde{\Lambda}_{1234...}=(\id_1\bigotimes_{i=2}^m\Omega_{(p_i)})(\Lambda_{1234...}),
\label{multisep}
\ee

\noindent where 
\be
p_i=\frac{d_i}{N+d_i}.
\ee

The corresponding expression for $\tilde{\Lambda}_{123...}$ when it arises from an $N$ locally Bose-symmetric extension, PPT with respect to the partition $12^{\lceil N/2\rceil} 3^{\lceil N/2\rceil}...|2^{\lfloor N/2\rfloor} 3^{\lfloor N/2\rfloor}...$, is still
(\ref{multisep}), but this time

\be
p_i=\frac{d_i}{2(d_i-1)}g_N^{(d_i)}.
\ee


\section{The power of PPT alone}
\label{sec: PPT}

The Peres-Horodecki criterion, aka the PPT (Positive Partial Transpose) criterion \cite{P96}, is one of the most popular existent criteria for entanglement detection. It is simple, it provides a very good approximation to the set of separable states in small dimensional cases and it usually leads to analytical results when applied over families of quantum states. Actually, some entanglement measures, like the negativity \cite{neg1} or the PPT entanglement robustness \cite{intro_measures} are based on the PPT condition.

It is interesting, thus, to try to determine how good the PPT criterion is for entanglement detection \emph{alone}, i.e., not in combination with Doherty et al.'s method. Here, through a very simple argument, we show what we believe is the first result in this direction after the seminal paper of the Horodeckis \cite{PPThoro}.

The main idea of our derivation stems from the fact that positivity under partial transposition is equivalent to separability in $\C^3\otimes\C^2$ systems \cite{PPThoro}. Suppose, then, that we have a PPT state $\rho_{AB}\in B(\H_A\otimes \H_B)$, with $d_A\geq 3$, and $d_B\geq 2$, and consider the (non normalized) state $\tilde{\rho}_{AB}$ given by

\be
\tilde{\rho}_{AB}\propto\int dUdW P_{U}^3\otimes P_{W}^2\rho_{AB}P_{U}^3\otimes P_{W}^2,
\label{decomp}
\ee

\noindent where $dU$ and $dW$ denote the Haar measures corresponding to $S(d_A)$ and $SU(d_B)$, respectively, and

\be
P_{U}^3\equiv U\sum_{k=0}^2\proj{k} U^\dagger,P_{W}^2\equiv W\sum_{k=0}^1\proj{k} W^\dagger.
\ee

It follows that $\tilde{\rho}_{AB}$ is a convex combination of unnormalized states $\rho_{U,W}\equiv P_{U}^3\otimes P_{W}^2\rho_{AB}P_{U}^3\otimes P_{W}^2$, with $\rho_{U,W}\in B(\C^3\otimes\C^2)$. Notice, also, that each $\rho_{U,W}$ is PPT, since

\begin{eqnarray}
\rho_{U,W}^{T_B}=& &(P_{U}^3\otimes P_{W}^2\rho_{AB}P_{U}^3\otimes P_{W}^2)^{T_B}=\nonumber\\
=& &P_{U}^3\otimes P_{W^*}^2\rho^{T_B}_{AB}P_{U}^3\otimes P_{W^*}^2\geq 0.
\end{eqnarray}

Since PPT equals separability in $\C^3\otimes \C^2$ systems, it follows that each $\rho_{U,W}$ is separable, and so is $\tilde{\rho}_{AB}$, since by construction it is a convex combination of these states.

It only rests to find an analytical expression for $\tilde{\rho}_{AB}$. Using the previous techniques it is straightforward to arrive at

\begin{theo}
Let $\rho_{AB}\in B(\H_A\otimes\H_B)$ be a PPT normalized quantum state, with $d_A\geq 3, d_B\geq 2$. Then, for

\be
p_A=\frac{d_A(d_A-3)}{d_A^2-1},p_B=\frac{d_B(d_B-2)}{d_B^2-1},
\ee

\noindent the state $\Omega^{(p_A)}\otimes\Omega^{(p_B)}(\rho_{AB})$ is separable.
\end{theo}

\noindent Note that, in the particular case $d_A=3,d_B=2$, $\tilde{\rho}_{AB}=\rho_{AB}$.

By simple application of the tools already developed, we end up with the following Corollary.

\begin{cor}
For any PPT state $\rho_{AB}$, with $d_A\geq 3,d_B\geq2$,

\be
R_G(\rho_{AB})\leq \frac{1}{12}(d_A+1)(d_B+1)-1,
\ee

\noindent and there exists a separable state $\sigma$ such that

\be
\|\rho_{AB}-\sigma\|_1\leq 2-\frac{24}{(d_A+1)(d_B+1)}.
\ee

\end{cor}

\begin{figure}
  \centering
  \includegraphics[width=8 cm]{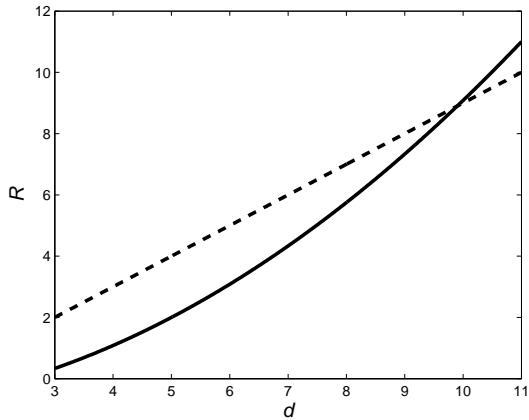}
  \caption{Optimum bound on the global robustness of entanglement $R$ for generic states (dashed line), as opposed to the upper bound for PPT states (solid line). In this plot, we assume that $d_A=d_B=d$. Note that the new bound becomes trivial as soon as $d>9$.}
  \label{PPT}
\end{figure}

To get an idea on how good these bounds are, have a look at Figure \ref{PPT}. There the maximum possible \emph{global} robustness of entanglement of a $\C^d\times \C^d$ state is compared with our upper bound for PPT states. We see that, although our upper bound becomes useless for $d>9$, it is very powerful in the small dimensional case. For instance, for $\C^3\times \C^3$ systems, the bound is equal to $1/3$ as opposed to $2$. This means that we would have to apply the non PPT version of the DPS method up to $N=6$ in order to characterize likewise the set of separable states.

\section{Conclusion}
\label{conclusion}
In this paper, we have studied the efficiency of the DPS criterion for entanglement detection. First, we showed that it is enough to subject the DPS states to some local noise in order to deprive them from their entanglement properties. It turned out that, while the minimal amount of noise necessary to turn an arbitrary state in $\bar{S}^N$ into a separable state decreases as $O(1/N)$, the corresponding amount of noise needed to disentangle states in $\bar{S}^N_p$ decreases at least as $O(1/N^2)$. We used these expressions to estimate the time complexity of both methods when applied to solve the Weak Membership Problem of Separability, and concluded that the PPT condition is worth imposing provided that the \emph{optimal} bounds on the speed of convergence of the method based on plain BSEs scale as $O(d/N)$, as our own bounds suggest. We therefore hope to have shed some light on the question of how much the DPS criterion owes its strength to the PPT condition. 

We also derived bounds on the error we incur when we substitute the set of separable operators by $S^N$ or $S^N_p$ in linear optimization problems, like the state estimation problem, the problem of determining the maximal output purity of an arbitrary quantum channel and the computation of the geometric entanglement. We performed numerical calculations of the first of these problems to test the accuracy of our analytical bounds. In order to compare our uncertainty with the actual solution of the problem, we developed a new technique that allows to prove in some cases the optimality of the DPS relaxations. We observed that, although the bounds for the non PPT case seem to be very accurate, the bounds for the PPT case are too big when compared with reality.

This disagreement between theory and practice may be explained in part by the fact that our bounds do not take into account the dimensionality of Alice's system, a crucial fact when dealing with the PPT constraint \cite{PPThoro}. For all we know, our PPT bounds could be exact in the limit $d_A\to \infty$. Our intuition, nevertheless, is that better bounds could be found by applying linear maps over the initial state $\rho_{AB}$ in order to obtain a separable state $\tilde{\rho}_{AB}$, as we did, but whose separable decomposition would be given by a \emph{non linear map}, unlike in our examples. Actually, we already used that approach in Section \ref{sec: PPT} to bound the entanglement of PPT states. That kind of schemes, together with state estimation considerations, may allow in the future to obtain such better bounds.

\section*{Acknowledgements}
The authors thank Animesh Datta and Fernando G. S. L. Brand\~{a}o for useful discussions. This work is part of the EPSRC QIP-IRC and is supported by EPSRC grant EP/C546237/1, the Royal Society, the EU Integrated Project QAP and an Institute for Mathematical Sciences postdoc fellowship.

\begin{appendix}

\section{Minimization of (\ref{complicaciones})}
\label{minimum}

Take $N$ even. Then it can be checked that

\begin{eqnarray}
& &\tilde{A}_{mn}=\int_0^1x^{m+n}\cdot\frac{(1-x)^{d-1}}{(d-1)!}dx,\nonumber\\
& &\tilde{B}_{mn}=\int_0^1x^{m+n}\cdot\frac{(1-x)^{d-2}}{(d-2)!}dx.
\label{idea1}
\end{eqnarray}

\noindent Combining this relation with (\ref{complicaciones}), it follows that

\be
f(\vec{c})=\frac{1}{d-1}\frac{\int_0^1|\sum_{n=0}^{N/2}c_nx^n|^2(1-x)(1-x)^{d-2}dx}{\int_0^1|\sum_{n=0}^{N/2}c_nx^n|^2(1-x)^{d-2}dx}.
\ee

That way, we can see the minimization of $f(\vec{c})$ as a minimization over
the set of all polynomials $Q_{N/2}(x)=\sum c_nx^n$ of degree $N/2$. Making the
change of coordinates $y=2x-1$ we find that the above minimization is equivalent
to

\be
\min_{Q_{N/2}}\frac{1}{2(d-1)}\frac{\int_{-1}^1|Q_{N/2}(y)|^2(1-y)^{d-1}dy}{\int_{-1}^1|Q_{N/2}(y)|^2(1-y)^{d-2}dy},
\ee

\noindent where $Q_{N/2}(y)$ is an arbitrary polynomial of order $N/2$. This
problem can be solved by means of the \emph{Jacobi polynomials}. 

The Jacobi polynomials $P_n^{(\alpha,\beta)}(y)$ are a complete set of functions orthogonal upon integration in the interval $[-1,1]$ under the weight $(1+y)^\beta(1-y)^\alpha$
\cite{abramo}. 
Now, define the \emph{normalized Jacobi polynomials} $p_n(y)$ as

\be
p_n(y)\equiv \frac{P^{(d-2,0)}_n(y)}{\|P^{(d-2,0)}_n\|}, 
\ee

\noindent with 

\be
\|P^{(d-2,0)}_n\|= \sqrt{\int_{-1}^1|P^{(d-2,0)}_n(y)|^2(1-y)^{d-2}dy}.
\ee

It is
clear that we can express any $Q_{N/2}(y)$ as a linear combination of normalized
Jacobi polynomials
of order less or equal than $N/2$. That is, 

\be
Q_{N/2}(y)=\sum_{n=0}^{N/2} e_np_n(y),
\ee

\noindent for some coefficients $e_n$. Because of the orthogonality of the
$p_n$'s, when we input this expression in the
integral of the denominator, we end up with

\be
\int_{-1}^1|Q_{N/2}(y)|^2(1-y)^{d-2}dy=\sum_n |e_n|^2.
\ee

To calculate the integral on the numerator, we can make use of the recurrence
relation

\be
(1-y)p_n(y)=\alpha_np_n(y)+\beta_np_{n+1}(y)+\gamma_np_{n-1}(y),
\label{recursion}
\ee

\noindent that holds for some coefficients $\alpha_n,\beta_n,\gamma_n$, with $\gamma_0=0$ and
$\gamma_{n+1}=\beta_n$ \cite{abramo}. Invoking again the orthogonality of the Jacobi
polynomials, we have that

\be
\min_{\vec{c}}f(\vec{c})=\min_{|\vec{e}|^2=1}\frac{1}{2(d-1)}\vec{e}^\dagger
\tilde{C}\vec{e},
\ee

\noindent where $\tilde{C}$ is an $(N/2+1)\times (N/2+1)$ tridiagonal hermitian matrix given by

\begin{eqnarray}
\tilde{C}_{m,n}=& &\alpha_n, \mbox{ if } m=n,\nonumber\\
& &\beta_n, \mbox{ if } m=n+1,\nonumber\\
& &\gamma_n, \mbox{ if } m=n-1,\nonumber\\
& &0 \mbox{ elsewhere}.
\end{eqnarray}

\noindent Now we will proceed to diagonalize $\tilde{C}$.

Let $\lambda$ be an eigenvalue
of $\tilde{C}$. This means that there exists a vector $\{v_i\}_{i=0}^{N/2+1}$ such that

\begin{equation}
(\alpha_n-\lambda)v_n+\beta_nv_{n+1}+\gamma_nv_{n-1}=0,
\label{eigenvector}
\end{equation}

\noindent with $v_{N/2+1}=0$.

Choose a real number $y_0$ and try the ansatz $v_n=p_n(y_0)$. From (\ref{recursion}), it is clear that $v_n$ will satisfy (\ref{eigenvector}),
provided that 

\begin{eqnarray}
&\lambda=1-y_0,\nonumber\\
&p_{N/2+1}(y_0)=0.
\end{eqnarray}

\noindent That is, any root of
the polynomial $p_{N/2+1}(y)$ corresponds to an eigenvalue of $\tilde{C}$.

But $p_{N/2+1}(y)$ has $N/2+1$ \emph{simple} roots \cite{abramo}, so all the eigenvalues of $\tilde{C}$ are obtained using this strategy. It follows that

\be
\min_{\vec{c}}f_N(\vec{c})=\frac{1}{2(d-1)}\min\{1-x:P_{N/2+1}^{(d-2,0)}(x)=0\}.
\ee

\noindent Let us remark that this is not the first time the zeros of the Jacobi polynomials naturally appear in state estimation problems \cite{jacobino}.

The expression for the case of odd $N$ can be derived in an analogous way taking into account that, this time,

\begin{eqnarray}
& &\tilde{A}_{mn}=\int_0^1x^{m+n}\cdot\frac{x(1-x)^{d-1}}{(d-1)!}dx,\nonumber\\
& &\tilde{B}_{mn}=\int_0^1x^{m+n}\cdot\frac{x(1-x)^{d-2}}{(d-2)!}dx.
\end{eqnarray}

\section{Optimality criterion (rank loops)}
\label{optim}
For some problems involving linear optimizations over the set $S$, it may happen (see \cite{pasado}) 
that a particular relaxation of the problem $F^N$ turns out to coincide with $F$. In this appendix we will show how this optimality can sometimes be detected.

We will take inspiration from optimality detection in other hierarchies of semidefinite programs that appear in scientific literature. 
Consider the hierarchy of semidefinite programs used in \cite{qcorr} for the calculation of the maximal violation of linear Bell inequalities. 
There the optimality of a relaxation is detected when the rank of the matrix generated by the computer is equal to that of some of its submatrices. 
Remarkably, we can find similar results in the hierarchies of semidefinite programs defined by Henrion and Lasserre 
to minimize real polynomials in a bounded region of $\R^n$ \cite{lasserre}. 

The corresponding result in this scenario is the following:

\begin{lemma}
\label{rankloop}
Let $\Lambda_{AB^N}$ be a BSE of $\Lambda_{AB}$, PPT with respect to the partition $AB^K|B^{N-K}$. If

\be
\mbox{rank}(\Lambda_{AB^N})\leq \max\{\mbox{rank}(\Lambda_{AB^K}),\mbox{rank}(\Lambda_{B^{N-K}})\}
\label{optimality}
\ee

\noindent then $\Lambda_{AB}$ is a separable operator.
\end{lemma}

Following \cite{qcorr}, we will say that $\Lambda_{AB^N}$ presents a \emph{rank loop} when it fulfills condition (\ref{optimality}).

The proof of Lemma \ref{rankloop} follows trivially from an old result by Horodecky et al. \cite{lowrank}:

\begin{theo}
Let $\rho_{AB}$ be a PPT bipartite quantum state. If

\be
\mbox{rank}(\rho_{AB})\leq \mbox{rank}(\rho_{A}),
\ee

\noindent then $\rho_{AB}$ is a separable state.
\end{theo}

\noindent See \cite{lowrank} for a proof.

The possibility of finding a rank loop in practice in cases where the optimization over the set $S_p^N$ coincides with the optimization over $S$ should not be surprising. Note that any (finite dimensional) separable state $\Lambda_{AB}$ can be expressed as a finite convex combination of product states, i.e.,

\be
\Lambda_{AB}=\sum_{i=1}^K p_i \rho_i\otimes \proj{\psi_i}, \mbox{ with } p_i>0,\forall i,
\ee

\noindent with $\proj{\psi_i}\not=\proj{\psi_j}$, for $i\not=j$. Now, consider the PPT Bose symmetric extension of $\Lambda_{AB}$ given by

\be
\Lambda_{AB^N}=\sum_{i=1}^K p_i \rho_i\otimes \proj{\psi_i}^{\otimes N},
\ee

Clearly, as $N$ tends to infinity, the vectors $\{\ket{\psi_i}^{\otimes N}\}_i$ become orthogonal. It follows that $K^*\equiv\lim_{N\to\infty}\mbox{rank}(\Lambda_{AB^N})$ exists and is equal to $\sum_i\mbox{rank}(\rho_i)$. Being the rank a natural number, this implies that there is an $M$ such that, for any $N> M$, $\mbox{rank}(\Lambda_{AB^M})=\mbox{rank}(\Lambda_{AB^N})=K^*$. That is, for any finite dimensional separable state there exists a PPT Bose symmetric extension with a rank loop.

Of course, the fact that for any separable state $\rho_{AB}$ there exists a PPT BSE with a rank loop does not mean that our computer is going to return such an extension. Note, though, that, if at the same time we set our computer to the task of finding PPT BSEs of $\rho_{AB}$ we also demand a rank minimization of these matrices (i.e., we look for PPT BSEs with minimal rank), at some point we will find a rank loop. 

Unfortunately, rank minimization of positive semidefinite matrices with linear constraints is in general an NP-hard problem \cite{PhD,sdp}. There are, however, heuristics \cite{maryam} that have proven to be very efficient for solving small-scale problems (that is, for small $d$).

\end{appendix}

\end{document}